\documentclass{article}

\usepackage[a4paper,hmargin=2.8cm,vmargin=3cm]{geometry}
\usepackage{graphicx}                   % graphics (includegraphics)
\usepackage{amsmath}                    % AMS math
\usepackage{amssymb}                    % special AMS math symbols
\usepackage{amsthm}                     % AMS theorem/proof
\usepackage{amsfonts}                   % more special AMS math symbols
\usepackage{url}                        % \url{...} formatting URLs
\usepackage{color}                      % colored text and backgrounds
\usepackage{cite}                       % sort bibliographical entries
\usepackage{hyperref}                   % add hyper-references
\hypersetup{colorlinks=true,linkcolor=[rgb]{0.75,0,0},citecolor=[rgb]{0,0,0.75}}%
\usepackage{wrapfig}                    % figures embedded in text

\newcommand{\RE}{\mathbb{R}}
\newcommand{\eps}{\varepsilon}
\newcommand{\half}[1]{\frac{#1}{2}}
\newcommand{\inv}[1]{\frac{1}{#1}}
\newcommand{\radius}{\mathrm{radius}}
\newcommand{\diam}{\mathrm{diam}}
\newcommand{\conv}{\mathrm{conv}}
\newcommand{\Sh}{\mathrm{shadow}}
\newcommand{\normals}{\mathrm{normals}}
\newcommand{\width}{\mathrm{width}}
\newcommand{\ray}{\mathrm{ray}}

\newtheorem{theorem}{Theorem}[section]
\newtheorem{lemma}[theorem]{Lemma}

\begin{document}
\title{Near-Optimal $\eps$-Kernel Construction and\protect\\ Related Problems}

\author{%
	Sunil Arya\thanks{Research supported by the Research Grants Council of Hong Kong, China under project number 16200014.}\\
		Department of Computer Science and Engineering \\
		The Hong Kong University of 
        Science and Technology \\
 		Clear Water Bay, Kowloon, Hong Kong\\
		arya@cse.ust.hk \\
		\and
	Guilherme D. da Fonseca\\
		Universit\'{e} Clermont Auvergne and\\
		LIMOS\\
 		Clermont-Ferrand, France\\
		fonseca@isima.fr
		\and
	David M. Mount\thanks{Research supported by NSF grant CCF--1618866.}\\
		Department of Computer Science and \\
		Institute for Advanced Computer Studies \\
		University of Maryland \\
 		College Park, Maryland 20742 \\
		mount@cs.umd.edu \\
}

\date{}

\maketitle

\begin{abstract}
The computation of (i) $\eps$-kernels, (ii) approximate diameter, and (iii) approximate  bichromatic closest pair are fundamental problems in geometric approximation. In this paper, we describe new algorithms that offer significant improvements to their running times. In each case the input is a set of $n$ points in $\RE^d$ for a constant dimension $d \geq 3$ and an approximation parameter $\eps > 0$. We reduce the respective running times\\
(i) from $O((n + 1/\eps^{d-2})\log\inv\eps)$ to $O(n \log\inv\eps + 1/\eps^{(d-1)/2+\alpha})$,\\
(ii) from $O((n + 1/\eps^{d-2})\log\inv\eps)$ to $O(n \log\inv\eps +  1/\eps^{(d-1)/2+\alpha})$, and\\
(iii) from $O(n / \eps^{d/3})$ to $O(n / \eps^{d/4+\alpha}),$\\
for an arbitrarily small constant $\alpha > 0$.
Result (i) is nearly optimal since the size of the output $\eps$-kernel is $\Theta(1/\eps^{(d-1)/2})$ in the worst case.

These results are all based on an efficient decomposition of a convex body using a hierarchy of Macbeath regions, and contrast to previous solutions that decompose space using quadtrees and grids. By further application of these techniques, we also show that it is possible to obtain near-optimal preprocessing time for the most efficient data structures to approximately answer queries for (iv) nearest-neighbor searching, (v) directional width, and (vi) polytope membership.
\end{abstract}

\section{Introduction} \label{s:intro}

In this paper we present new faster algorithms to several fundamental geometric approximation problems involving point sets in $d$-dimensional space. In particular, we present approximation algorithms for $\eps$-kernels, diameter, bichromatic closest pair, and the minimum bottleneck spanning tree. Our results arise from a recently developed shape-sensitive approach to approximating convex bodies, which is based on the classical concept of Macbeath regions. This approach has been applied to computing area-sensitive bounds for polytope approximation~\cite{AFM12b}, polytope approximations with low combinatorial complexity~\cite{AFM16}, answering approximate polytope-membership queries~\cite{AFM17a}, and approximate nearest-neighbor searching~\cite{AFM17a}. The results of \cite{AFM17a} demonstrated the existence of data structures for these query problems but did not discuss preprocessing in detail. We complete the story by presenting efficient algorithms for building data structures for three related queries: approximate polytope membership,  approximate directional width, and approximate nearest-neighbors.

Throughout, we assume that the dimension $d$ is a constant. Our running times will often involve expressions of the form $1/\eps^\alpha$. In such cases, $\alpha > 0$ is constant that can be made arbitrarily small. The approximation parameter $\eps$ is treated as an asymptotic variable that approaches $0$. We assume throughout that $\eps < 1$, which guarantees that $\log\inv\eps > 0$.

In Section~\ref{ss:static}, we present our results for $\eps$-kernels, diameter, bichromatic closest pair, and minimum bottleneck tree. In Section~\ref{ss:ds}, we present our results for the data structure problems. In Section~\ref{ss:techniques}, we give an overview of the techniques used.

Concurrently and independently, Timothy Chan has reported complexity bounds that are very similar to our results~\cite{Cha17}. Remarkably, the computational techniques seem to be very different, based on Chebyshev polynomials.

\subsection{Static Results} \label{ss:static}

\subparagraph{Kernel.}
Given a set $S$ of $n$ points in $\RE^d$ and an approximation parameter $\eps>0$, an \emph{$\eps$-coreset} is an (ideally small) subset of $S$ that approximates some measure over $S$ (see~\cite{AHV05} for a survey). Given a nonzero vector $v \in \RE^d$, the \emph{directional width} of $S$ in direction $v$, $\width_v(S)$ is the minimum distance between two hyperplanes that enclose $S$ and are orthogonal to $v$. A \emph{coreset for the directional width} (also known as an \emph{$\eps$-kernel} and as a \emph{coreset for the extent measure}) is a subset $Q \subseteq S$ such that $\width_v(Q) \geq (1-\eps)\, \width_v(S)$, for all $v \in \RE^d$. Kernels are among the most fundamental constructions in geometric approximation, playing a role similar to that of convex hulls in exact computations. Kernels have been used to obtain approximation algorithms to several problems such as diameter, minimum width, convex hull volume, minimum enclosing cylinder, minimum enclosing annulus, and minimum-width cylindrical shell~\cite{AHV04,AHV05}.

The concept of $\eps$-kernels was introduced by Agarwal et al.~\cite{AHV04}. The existence of $\eps$-kernels with $O(1/\eps^{(d-1)/2})$ points is implied in the works of Dudley~\cite{Dud74} and Bronshteyn and Ivanov~\cite{BrI76}, and this is known to be optimal in the worst case. Agarwal et al.~\cite{AHV04} demonstrated how to compute such a kernel in $O(n + 1/\eps^{3(d-1)/2})$ time, which reduces to $O(n)$ when $n = \Omega(1/\eps^{3(d-1)/2})$. While less succinct $\eps$-kernels with $O(1/\eps^{d-1})$ points can be constructed in time $O(n)$ time for all $n$~\cite{AHV04,BFP82}, no linear-time algorithm is known to build an $\eps$-kernel of optimal size. Hereafter, we use the term $\eps$-kernel to refer exclusively to an $\eps$-kernel of size $O(1/\eps^{(d-1)/2})$.

Chan~\cite{Cha06} showed that an $\eps$-kernel can be constructed in $O((n + 1/\eps^{d-2})\log \inv{\eps})$ time, which is nearly linear when $n = \Omega(1/\eps^{d-2})$. He posed the open problem of obtaining a faster algorithm.
A decade later, Arya and Chan~\cite{ArC14} showed how to build an $\eps$-kernel in roughly $O(n + \sqrt{n}/\eps^{d/2})$ time using discrete Voronoi diagrams. In this paper, we attain the following near-optimal construction time.

\begin{theorem} \label{thm:kernel}
Given a set $S$ of $n$ points in $\RE^d$ and an approximation parameter $\eps > 0$, it is possible to construct an $\eps$-kernel of $S$ with $O(1/\eps^{(d-1)/2})$ points in $O(n \log\inv\eps + 1/\eps^{(d-1)/2+\alpha})$ time, where $\alpha$ is an arbitrarily small positive constant.
\end{theorem}

We note that when $n = o(1/\eps^{(d-1)/2})$, the input $S$ is already an $\eps$-kernel and therefore an $O(n)$ time algorithm is trivial. Because the worst-case output size is $O(1/\eps^{(d-1)/2})$, we may assume that $n$ is at least this large, for otherwise we can simply take $S$ itself to be the kernel. Since $1/\eps^{\alpha}$ dominates $\log\inv\eps$, the above running time can be expressed as $O(n/\eps^\alpha)$, which is nearly linear given that $\alpha$ can be made arbitrarily small. 

\subparagraph{Diameter.}
An important application of $\eps$-kernels is to approximate the diameter of a point set. Given $n$ data points, the \emph{diameter} is defined to be the maximum distance between any two data points. An \emph{$\eps$-approximation} of the diameter is a pair of points whose distance is at least $(1-\eps)$ times the exact diameter. There are multiple algorithms to approximate the diameter~\cite{AHV04,AMS92,ArC14,BaH01,Cha06}. The fastest running times are $O((n + 1/\eps^{d-2})\log \inv{\eps})$~\cite{Cha06} and roughly $O(n + \sqrt{n}/\eps^{d/2})$~\cite{ArC14}. The algorithm from~\cite{Cha06} essentially computes an $\eps$-kernel $Q$ and then determines the maximum value of $\width_v(Q)$ among a set of $k = O(1/\eps^{(d-1)/2})$ directions $v$ by brute force~\cite{AHV04}. Discrete Voronoi diagrams~\cite{ArC14} permit this computation in roughly $O(n + \sqrt{n}/\eps^{d/2})$ time. Therefore, combining the kernel construction of Theorem~\ref{thm:kernel} with discrete Voronoi diagrams~\cite{ArC14}, we reduce $n$ to $O(1/\eps^{(d-1)/2})$ and obtain an algorithm to $\eps$-approximate the diameter in roughly $O(n + 1/\eps^{3d/4})$ time. However, we show that it is possible to obtain a much faster algorithm, as presented in the following theorem.

\begin{theorem} \label{thm:diameter}
Given a set $S$ of $n$ points in $\RE^d$ and an approximation parameter $\eps > 0$, it is possible to compute an $\eps$-approximation to the diameter of $S$ in $O(n \log\inv\eps + 1/\eps^{(d-1)/2+\alpha})$ time.
\end{theorem}

\subparagraph{Bichromatic Closest Pair.}
In the \emph{bichromatic closest pair} (BCP) problem, we are given $n$ points from two sets, designated red and blue, and we want to find the closest red-blue pair. In the \emph{$\eps$-approximate} version, the goal is to find a red-blue pair of points whose distance is at most $(1+\eps)$ times the exact BCP distance. Approximations to the BCP problem were introduced in~\cite{KhM95}, and the most efficient randomized approximation algorithm runs in roughly $O(n/\eps^{d/3})$ expected time~\cite{ArC14}. We present the following result.

\begin{theorem} \label{thm:bcp}
Given $n$ red and blue points in $\RE^d$ and an approximation parameter $\eps > 0$, there is a randomized algorithm that computes an $\eps$-approximation to the bichromatic closest pair in $O(n / \eps^{d/4+\alpha})$ expected time.
\end{theorem}

\subparagraph{Euclidean Trees.}
Given a set $S$ of $n$ points in $\RE^d$, a \emph{Euclidean minimum spanning tree} is the spanning tree with vertex set $S$ that minimizes the sum of the edge lengths, while a \emph{Euclidean minimum bottleneck tree} minimizes the maximum edge length. In the approximate version we respectively approximate the sum and the maximum of the edge lengths. A minimum spanning tree is a minimum bottleneck tree (although the converse does not hold). However, an approximation to the minimum spanning tree is not necessarily an approximation to the minimum bottleneck tree. A recent approximation algorithm to the Euclidean minimum spanning tree takes roughly $O(n \log n + n / \eps^2)$ time, regardless of the (constant) dimension~\cite{ArM16}. On the other hand, the fastest algorithm to approximate the minimum bottleneck tree 
takes roughly $O((n \log n)/\eps^{d/3})$ expected time~\cite{ArC14}.The algorithm uses BCP to simultaneously attain an approximation to the minimum bottleneck and the minimum spanning trees.
We prove the following theorem.

\begin{theorem} \label{thm:mst}
Given $n$ points in $\RE^d$ and an approximation parameter $\eps > 0$, there is a randomized algorithm that computes a tree $T$ that is an $\eps$-approximation to both the Euclidean minimum bottleneck and the Euclidean minimum spanning trees in $O((n \log n)/ \eps^{d/4+\alpha})$ expected time.
\end{theorem}

\subsection{Data Structure Results} \label{ss:ds}

\subparagraph{Polytope membership.}
Let $P$ denote a convex polytope in $\RE^d$, represented as the bounded intersection of $n$ halfspaces. The \emph{polytope membership problem} consists of preprocessing $P$ so that it is possible to determine efficiently whether a given query point $q \in \RE^d$ lies within $P$. In the $\eps$-approximate version, we consider an expanded convex body $K \supset P$. A natural way to define this expansion would be to consider the set of points that lie within distance $\eps \cdot \diam(P)$ of $P$, thus defining a body whose Hausdorff distance from $P$ is $\eps \cdot \diam(P)$. However, this definition has the shortcoming that it is not sensitive to the directional width of $P$. Instead, we define $K$ as follows. For any nonzero vector $v \in \RE^d$, consider the two supporting hyperplanes for $P$ that are normal to $v$. Translate each of these hyperplanes outward by a distance of $\eps \cdot \width_v(P)$, and consider the closed slab-like region lying between them. Define $K$ to be the intersection of this (infinite) set of slabs. This is clearly a stronger approximation than the Hausdorff-based definition. An \emph{$\eps$-approximate polytope membership query} ($\eps$-APM query) returns a positive result if the query point $q$ is inside $P$, a negative result if $q$ is outside $K$, and may return either result otherwise.\footnote{Our earlier works on $\eps$-APM queries~\cite{AFM11,AFM17a} use the weaker Hausdorff form to define the problem, but the solutions presented there actually achieve the stronger direction-sensitive form.}

We recently proposed an optimal data structure to answer approximate polytope membership queries, but efficient preprocessing remained an open problem~\cite{AFM17a}. In this paper, we present a similar data structure that not only attains optimal storage and query time, but can also be preprocessed in near-optimal time.

\begin{theorem} \label{thm:apm}
Given a convex polytope $P$ in $\RE^d$ represented as the intersection of $n$ halfspaces and an approximation parameter $\eps > 0$, there is a data structure that can answer $\eps$-approximate polytope membership queries with
\[
\hbox{Query time:~} O\kern-2pt \left( \log \inv{\eps} \right)\,
\textrm{Space:~} O\kern-2pt \left(1/\eps^{\half{d-1}}\right)\,
\textrm{Preprocessing:~} O\kern-2pt \left( n \log \inv \eps + \left(\inv\eps\right)^{\kern-2pt \half{d-1}+\alpha} \right).
\] 
\end{theorem}

\subparagraph{Directional width.}
Applying the previous data structure in the dual space, we obtain a data structure for the following \emph{$\eps$-approximate directional width} problem, which is closely related to $\eps$-kernels. Given a set $S$ of $n$ points in a constant dimension $d$ and an approximation parameter $\eps>0$, the goal is to preprocess $S$ to efficiently $\eps$-approximate $\width_v(S)$, for a nonzero query vector $v$. We present the following result. 

\begin{theorem} \label{thm:widthqueries}
Given a set $S$ of $n$ points in $\RE^d$ and an approximation parameter $\eps > 0$, there is a data structure that can answer $\eps$-approximate directional width queries with
\[
\hbox{Query time:~} O\kern-2pt \left( \log^2 \inv{\eps} \right)\,
\textrm{Space:~} O\kern-2pt \left(1/\eps^{\half{d-1}}\right)\,
\textrm{Preprocessing:~} O\kern-2pt \left( n \log \inv \eps + \left(\inv\eps\right)^{\kern-2pt \half{d-1}+\alpha} \right).
\] 
\end{theorem}

\subparagraph{Nearest Neighbor.}
Let $S$ be a set of $n$ points in $\RE^d$. Given any $q \in \RE^d$, an \emph{$\eps$-approximate nearest neighbor} (ANN) of $q$ is any point of $S$ whose distance from $q$ is at most $(1+\eps)$ times the distance to $q$'s closest point in $S$. The objective is to preprocess $S$ in order to answer such queries efficiently. Data structures for approximate nearest neighbor searching (in fixed dimensions) have been proposed by several authors, offering space-time tradeoffs (see~\cite{AFM17a} for an overview of the tradeoffs). Applying the reduction from approximate nearest neighbor to approximate polytope membership established in~\cite{AFM11} together with Theorem~\ref{thm:apm}, we obtain the following result, which matches the best bound~\cite{AFM17a} up to an $O(\log \inv \eps)$ factor in the query time, but offers faster preprocessing time. 

\begin{theorem} \label{thm:ann}
Given a set $S$ of $n$ points in $\RE^d$, an approximation parameter $\eps > 0$, and $m$ such that $\log \inv{\eps} \leq m \leq 1/(\eps^{d/2}\log\inv\eps)$, there is a data structure that can answer Euclidean $\eps$-approximate nearest neighbor queries with
\[
\hbox{Query time:~} O\kern-2pt \left(\log n + \frac{\log \inv\eps}{m \cdot \eps^{\half d} }\right) \vspace{0.8em}
\textrm{Space:~} O\big(n \kern+1pt m \big) \vspace{0.8em}
\textrm{Preprocessing:~} O\kern-2pt \left( n \log n \kern+1pt \log \inv\eps + \frac{n \kern+1pt m}{\eps^{\alpha}} \right).
\] 
\end{theorem}

\subsection{Techniques} \label{ss:techniques}

In contrast to previous kernel constructions, which are based on grids and the execution of Bronshteyn and Ivanov's algorithm, our construction employs a classical structure from the theory of convexity, called \emph{Macbeath regions}~\cite{Mac52}. Macbeath regions have found numerous uses in the theory of convex sets and the geometry of numbers (see B\'{a}r\'{a}ny~\cite{Bar00} for an excellent survey). They have also been applied to several problems in the field of computational geometry. However, most previous results were either in the form of lower bounds~\cite{AMM09b, AMX12, BCP93} or focused on existential results~\cite{AFM12b, AFM16, DGJ17, MuR14}.

In \cite{AFM17a} the authors introduced a data structure based on a hierarchy of ellipsoids based on Macbeath regions to answer approximate polytope membership queries, but the efficient computation of the hierarchy was not considered. In this paper, we show how to efficiently construct the Macbeath regions that form the basis of this hierarchy.

Let $P$ denote a convex polytope in $\RE^d$. Each level $i$ in the hierarchy corresponds to a $\delta_i$-approximation of the boundary of $P$ by a set of $O(1/\delta_i^{(d-1)/2})$ ellipsoids, where $\delta_i = \Theta(1/2^i)$. Each ellipsoid is sandwiched between two Macbeath regions and has $O(1)$ children, which correspond to the ellipsoids of the following level that approximate the same portion of the boundary (see Figure~\ref{f:hierarchy}). The hierarchy starts with $\delta_0 = \Theta(1)$ and stops after $O(\log\inv\delta)$ levels when $\delta_i = \delta$, for a desired approximation $\delta$. We present a simple algorithm to construct the hierarchy in $O(n + 1/\delta^{3(d-1)/2})$ time. The polytope $P$ can be presented as either the intersection of $n$ halfspaces or the convex hull of $n$ points. We present the relevant background in Section~\ref{s:hierarchy}.

\begin{figure}[htbp]
\centerline{\includegraphics[scale=.75,clip=true,trim=0 0 0 18]{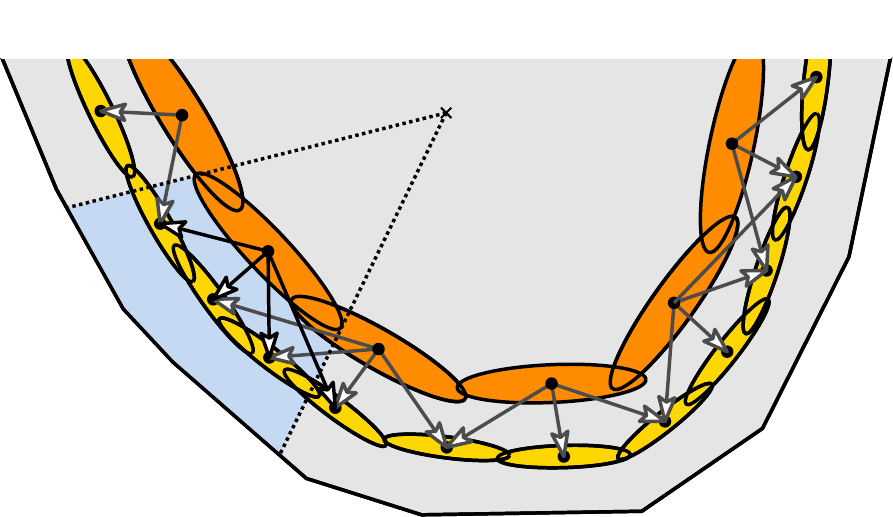}}
\caption{\label{f:hierarchy}Two levels of the hierarchy of ellipsoids based on Macbeath regions.}
\end{figure}

Our algorithm to compute an $\eps$-kernel in time $O(n \log \inv \eps + 1/\eps^{(d-1)/2+\alpha})$  (Theorem~\ref{thm:kernel}) is conceptually quite simple. Since the time to build the $\eps$-approximation hierarchy for the convex hull is prohibitively high, we use an approximation parameter $\delta = \eps^{1/3}$ to build a $\delta$-approximation hierarchy in $O(n + 1/\eps^{(d-1)/2})$ time. By navigating through this hierarchy, we partition the $n$ points among the leaf Macbeath ellipsoids in $O(n \log \inv\eps)$ time, discarding points that are too far from the boundary. We then compute an $(\eps/\delta)$-kernel for the set of points in each leaf ellipsoid and return the union of the kernels computed.

Given an algorithm to compute an $\eps$-kernel in $O(n \log \inv \eps + 1/\eps^{t(d-1)})$ time, the previous procedure produces an $\eps$-kernel in $O(n \log \inv \eps + 1/\eps^{t'(d-1)})$ time where $t' = (4t+1)/6$. Bootstrapping the construction a constant number of times, the value of $t$ goes down from $1$ to a value that is arbitrarily close to $1/2$. This discrepancy accounts for the $O(1/\eps^{\alpha})$ factors in our running times. In Section~\ref{s:kernel}, we present the complete algorithm and its analysis, proving Theorem~\ref{thm:kernel}.

In Section~\ref{s:apm}, we use our kernel construction in the dual space to efficiently build a polytope membership data structure, proving Theorem~\ref{thm:apm}. The key idea is to compute multiple kernels in order to avoid examining the whole polytope when building each Macbeath region. Again, we use bootstrapping to obtain a near-optimal preprocessing time. The remaining theorems follow from Theorems~\ref{thm:kernel} and~\ref{thm:apm}, together with several known
reductions (Section~\ref{s:reductions}). 

\section{Geometric Preliminaries} \label{s:prelim}

Consider a convex body $K$ in $d$-dimensional space $\RE^d$. Let $\partial K$ denote the boundary of $K$. Let $O$ denote the origin of $\RE^d$. Given a parameter $0 < \gamma \le 1$, we say that $K$ is \emph{$\gamma$-fat} if there exist concentric Euclidean balls $B$ and $B'$, such that $B \subseteq K \subseteq B'$, and $\radius(B) / \radius(B') \ge \gamma$. We say that $K$ is \emph{fat} if it is $\gamma$-fat for a constant $\gamma$ (possibly depending on $d$, but not on $\eps$). 

Unless otherwise specified, the notion of \emph{$\eps$-approximation} between convex bodies will be based on the direction-sensitive definition given in Section~\ref{ss:ds}. We say that a convex body $K'$ is an \emph{absolute $\eps$-approximation} to another convex body $K$ if they are within Hausdorff error $\eps$ of each other. Further, we say that $K'$ is an \emph{inner} (resp., \emph{outer}) approximation if $K' \subseteq K$ (resp., $K' \supseteq K$).

Let $B_0$ denote a ball of radius $r_0 = 1/2$ centered at the origin. For $0 < \gamma \le 1$, let $\gamma B_0$ denote the concentric ball of radius $\gamma r_0 = \gamma / 2$. We say that a convex body $K$ is in \emph{$\gamma$-canonical form} if it is nested between $\gamma B_0$ and $B_0$. A body in $\gamma$-canonical form is $\gamma$-fat and has diameter $\Theta(1)$. We will refer to point $O$ as the \emph{center} of $P$.

For any point $x \in K$, define $\delta(x)$ to be minimum distance from $x$ to any point on $\partial K$. For the sake of ray-shooting queries, it is useful to define a ray-based notion of distance as well. Given $x \in K$, define the \emph{ray-distance} of $x$ to the boundary, denoted $\ray(x)$, as follows. Consider the intersection point $p$ of $\partial K$ and the ray emanating from $O$ and passing through $x$. We define $\ray(x) = \|xp\|$. The following utility lemma will be helpful in relating distances to the boundary.

\begin{lemma} \label{lem:raydist}
Given a convex body $K$ in $\gamma$-canonical form:
\begin{enumerate}
\item[$(a)$\hspace{-5pt}] For any point $x \in P$, $\delta(x) \le \ray(x) \le \delta(x) / \gamma$.

\item[$(b)$\hspace{-4pt}] Let $h$ be a supporting hyperplane of $K$. Let $p$ be any point inside $K$ at distance at most $w$ from $h$, where $w \le  \gamma / 4$. Let $p'$ denote the intersection of the ray $Op$ and $h$. Then $\|pp'\| \le 2w /  \gamma$.

\item[$(c)$\hspace{-4pt}] Let $p$ be any point on the boundary of $K$, and let $h$ be a supporting hyperplane at $p$. Let $h'$ denote the hyperplane obtained by translating $h$ in the direction of the outward normal by $w$. Let $p'$ denote the intersection of the ray $Op$ with $h'$. Then $\|pp'\| \le w / \gamma$.
\end{enumerate}
\end{lemma}

We omit the straightforward proof. The lower bound on $\ray(x)$ for part~(a) is trivial, and the upper bound follows by a straightforward adaption of Lemma 4.2 of \cite{AFM16}. Part~(b) is an adaptation of Lemma~{2.11} of \cite{AFM17a}, and part~(c) is similar.

For any centrally symmetric convex body $A$, define $A^{\lambda}$ to be the body obtained by scaling $A$ by a factor of $\lambda$ about its center. The following lemma appears in Barany~\cite{Bar89}.

\begin{lemma} \label{lem:sym}
Let $\lambda \ge 1$. Let $A$ and $B$ be centrally symmetric convex bodies such that $A \subseteq B$. Then
$A^{\lambda} \subseteq B^{\lambda}$.
\end{lemma}

\subsection{Caps and Macbeath Regions}

Much of the material in this section has been presented in~\cite{AFM16,AFM17a}. We include it here for the sake of completeness. Given a convex body $K$, a \emph{cap} $C$ is defined to be the nonempty intersection of $K$ with a halfspace (see Figure~\ref{f:cap}(a)).
Let $h$ denote the hyperplane bounding this halfspace. We define the \emph{base} of $C$ to be $h \cap K$. The \emph{apex} of $C$ is any point in the cap such that the supporting hyperplane of $K$ at this point is parallel to $h$. The \emph{width} of $C$, denoted $\width(C)$, is the distance between $h$ and this supporting hyperplane. Given any cap $C$ of width $w$ and a real $\lambda \ge 0$, we define its \emph{$\lambda$-expansion}, denoted $C^{\lambda}$, to be the cap of $K$ cut by a hyperplane parallel to and at distance $\lambda w$ from this supporting hyperplane. (Note that $C^{\lambda} = K$, if $\lambda w$ exceeds the width of $K$ along the defining direction.) 

\begin{figure}[htbp]
\centerline{\includegraphics[scale=.65]{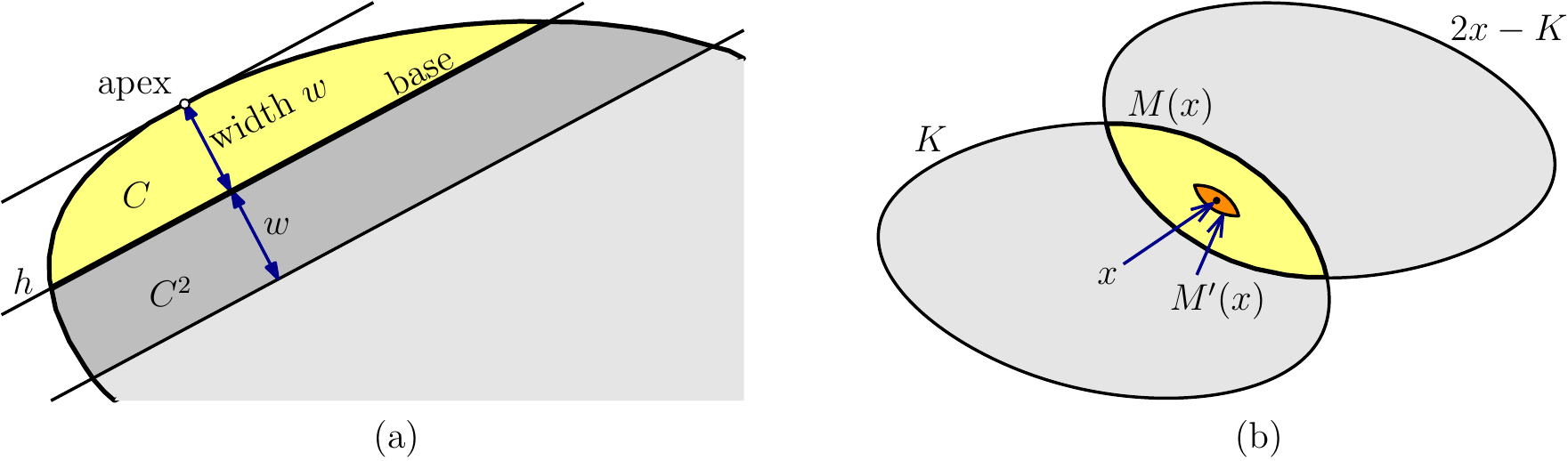}}
\caption{\label{f:cap}(a) Cap concepts and (b) Macbeath regions.}
\end{figure}

Given a point $x \in K$ and real parameter $\lambda \ge 0$, the \emph{Macbeath region} $M^{\lambda}(x)$ (also called an \emph{M-region}) is defined as:
\[
M^{\lambda}(x) ~ = ~ x + \lambda ((K-x) \cap (x-K)).
\]
It is easy to see that $M^{1}(x)$ is the intersection of $K$ and the reflection of $K$ around $x$ (see Figure~\ref{f:cap}(b)). Clearly, $M^{1}(x)$ is centrally symmetric about $x$, and $M^{\lambda}(x)$ is a scaled copy of $M^{1}(x)$ by the factor $\lambda$ about $x$. We refer to $x$ as the \emph{center} of $M^{\lambda}(x)$ and to $\lambda$ as its \emph{scaling factor}. As a convenience, we define $M(x) = M^1(x)$ and $M'(x) = M^{1/5}(x)$. We refer to the latter as the \emph{shrunken} Macbeath region.

We now present a few lemmas that encapsulate key properties of Macbeath regions. The first lemma shows that if two shrunken Macbeath regions have a nonempty intersection, then a constant factor expansion of one contains the other~\cite{AFM17a,BCP93,ELR70}.

\begin{lemma} \label{lem:mac-mac}
Let $K$ be a convex body, and let $\lambda \le 1/5$ be any real. If $x,y \in K$ such that $M^{\lambda}(x) \cap M^{\lambda}(y) \neq \emptyset$, then $M^{\lambda}(y) \subseteq M^{4\lambda}(x)$.
\end{lemma}

The next lemma is useful in situations when we know that a shrunken Macbeath region partially overlaps a cap of $K$. It allows us to conclude that a constant factor expansion of the cap will fully contain the Macbeath region. The proof appears in \cite{AFM16}.

\begin{lemma} \label{lem:mac-cap2}
Let $K$ be a convex body. Let $C$ be a cap of $K$ and $x$ be a point in $K$ such that $C \cap M'(x) \neq \emptyset$.
Then $M'(x) \subseteq C^2$.
\end{lemma}

The following lemma shows that all points in a shrunken Macbeath region have similar distances from the boundary of $K$. The proof appears in \cite{AFM17a}.

\begin{lemma} \label{lem:core-delta}
Let $K$ be a convex body. If $x \in K$ and $x' \in M'(x)$, then $4 \delta(x) / 5 \le \delta(x') \le 4 \delta(x) / 3$.
\end{lemma}

For any $\delta > 0$, define the $\delta$-erosion of a convex body $K$, denoted $K(\delta)$, to be the closed convex body formed by removing from $K$ all points lying within distance $\delta$ of $\partial K$. The next lemma bounds the number of disjoint Macbeath regions that can be centered on the boundary of $K(\delta)$. The proof appears in \cite{AFM17a}.

\begin{lemma} \label{lem:bound-mac}
Consider a convex body $K \subset \RE^d$ in $\gamma$-canonical form for some constant $\gamma$. Define $\Delta_0 = \frac{1}{2} (\gamma^2 / (4d))^d$. For any fixed constant $0 < \lambda \le 1/5$ and real parameter $\delta \le \Delta_0$, let $\mathcal{M}$ be a set of disjoint $\lambda$-scaled Macbeath regions whose centers lie on the boundary of $K(\delta)$. Then $|\mathcal{M}| = O(1/\delta^{(d-1)/2})$.
\end{lemma}

\subsection{Shadows of Macbeath regions}

Shrunken Macbeath regions reside within the interior of the convex body, but it is useful to identify the portion of the body's boundary that this Macbeath region will be responsible for approximating. For this purpose, we introduce the shadow of a Macbeath region. Given a convex body $K$ that contains the origin $O$ and a region $R \subseteq K$, we define the \emph{shadow} of $R$ (with respect to $K$), denoted $\Sh(R)$, to be the set of points $x \in K$ such that the line segment $Ox$ intersects $R$. 

We also define a set of \emph{normal directions} for $R$, denoted $\normals(R)$. Consider the set of all hyperplanes that support $K$ at some point in the shadow of $R$. Define $\normals(R)$ to be the set of outward unit normals to these supporting hyperplanes. Typically, the region $R$ in our constructions will be a (scaled) Macbeath region or an associated John ellipsoid (as defined in Section~\ref{s:hierarchy}), close to the boundary of $K$. The following lemma captures a salient feature of these shadows, namely, that the shadow of a Macbeath region $M'(x)$ can be enclosed in an ellipsoid whose width in all normal directions is $O(\delta(x))$.

\begin{lemma} \label{lem:shadow}
Let $K \subset \RE^d$ be a convex body in $\gamma$-canonical form for some constant $\gamma$. Let $x$ be a point at distance $\delta$ from the boundary of $K$, where $\delta \le \Delta_0$. Let $M = M'(x), S = \Sh(M)$, $N = \normals(M)$, and $\widehat{M} = M^{4/\gamma}(x)$. Then:
\begin{enumerate}
\item[$(a)$\hspace{-5pt}] $S \subseteq \widehat{M}$.
\item[$(b)$\hspace{-4pt}] $\width_v(S) \le  c_1 \delta$ for all $v \in N$. Here $c_1$ is the constant $8 / (3 \gamma)$.  
\item[$(c)$\hspace{-4pt}] $\width_v(\widehat{M}) \le c_2 \delta$ for all $v \in N$. Here $c_2$ is the constant $160 / (3 \gamma^2)$.
\end{enumerate}
\end{lemma}

\begin{proof}
We first prove (a). Consider any point $p \in \partial K \cap S$. Let $y$ denote the first point of intersection of the ray $Op$ with the Macbeath region $M$. To prove (a), it suffices to show that the segment $yp$ is contained in $\widehat{M}$ which, by convexity, is equivalent to showing that both points $y$ and $p$ are contained in $\widehat{M}$. Since $y \in M$, we have $y \in \widehat{M}$. To prove that $p \in \widehat{M}$, observe that a straightforward consequence of the definition of Macbeath regions is that $M(y)$ must contain a ball of radius $\delta(y)$ centered at $y$. Further, by Lemma~\ref{lem:raydist}(a), $\|yp\| = \ray(y) \le \delta(y) / \gamma$. It follows that $p \in M^{1/\gamma}(y)$. Also, since $y \in M'(x)$, it follows trivially that $M'(y)$ overlaps with $M'(x)$. Thus, by Lemma~\ref{lem:mac-mac}, $M'(y) \subseteq M^{4/5}(x)$. Applying Lemma~\ref{lem:sym} to $M'(y)$ and $M^{4/5}(x)$ with $\lambda = 5/\gamma$, we obtain $M^{1/\gamma}(y) \subseteq M^{4/\gamma}(x) = \widehat{M}$. Thus $p \in \widehat{M}$, which proves (a).

To prove (b), let $p$ be any point of $\partial K \cap S$ and let $y$ denote any point in the intersection of the ray $Op$ with $M'(x)$. Let $h$ denote any hyperplane supporting $K$ at $p$. Let $v$ denote the outward normal to $h$. We translate hyperplane $h$ to pass through $y$ and let $C$ denote the cap of $K$ cut by the resulting hyperplane. Clearly $\width(C) \le \|py\|$. By Lemma~\ref{lem:raydist}(a), $\|py\| = \ray(y) \le \delta(y) / \gamma$. Since $y \in M'(x)$, by Lemma~\ref{lem:core-delta}, $\delta(y) \le 4 \delta(x) / 3 = 4 \delta / 3$. Thus $\width(C) \le \delta(y) / \gamma \le 4 \delta / (3 \gamma)$. Also, by Lemma~\ref{lem:mac-cap2}, since $M'(x)$ overlaps $C$, $M'(x) \subseteq C^2$. It follows from convexity that $S \subseteq C^2$ and thus 
\[
	\width_v(S) 
		~ \le ~ \width(C^2) 
		~ \le ~ 2 \, \width(C) 
		~ \le ~ 8 \delta / (3 \gamma),
\]
which proves (b).

To prove (c), recall that $M'(x) \subseteq S$, which implies that $\width_v(M'(x)) \le 8 \delta / (3 \gamma)$. Thus $\width_v(\widehat{M}) = (20/\gamma) \width_v(M'(x)) \le 160 \kern+1pt \delta / (3 \gamma^2)$.
\end{proof}

\subsection{Representation Conversions}

Convex sets are naturally described in two ways, as the convex hull of a discrete set of points and as the intersection of a discrete set of halfspaces. Some computational tasks are more easily performed using one operation or the other. For this reason, it will be useful to be able to convert between one representation and the other. Also, when approximate representations suffice, it will be useful to prune a large set down to a smaller size. In this section we will present a few technical utilities to perform these conversions. 

Given an $n$-element point set in $\RE^d$, Chan showed that it is possible to construct an $\eps$-kernel of size $O(1/\eps^{(d-1)/2})$ in time $O(n + 1/\eps^{d-1})$~\cite{Cha06}. The following lemma shows that, by applying Chan's construction, is is possible to efficiently approximate the convex hull of $n$ points as the intersection of halfspaces.

\begin{lemma} \label{lem:conversion}
Let $\gamma < 1$ be a positive constant, and $\eps > 0$ be a real parameter. Let $P$ be a polytope in $\gamma$-canonical form represented as the convex hull of $n$ points. In $O(n + 1/\eps^{d-1})$ time it is possible to compute a polytope $P'$ represented as the intersection of $O(1/\eps^{(d-1)/2})$ halfspaces such that $P'$ is an inner absolute $\eps$-approximation of $P$.
\end{lemma}

\begin{proof}
Throughout the proof, to avoid tracking the numerous constant factors, we use the notation $O(\eps)$ to denote a quantity that is a suitable scaling of $\eps$ by a constant factor. Let $S$ be the input set of $n$ points such that $P = \conv(S)$. Since $P$ is in $\gamma$-canonical form, we have $S \subseteq B_0$, where $B_0$ denotes the ball of radius $r_0 = 1/2$ centered at the origin. By applying Chan's kernel construction~\cite{Cha06}, in time $O(n + 1/\eps^{d-1})$ we can compute a point set $S''$ of size $O(1/\eps^{(d-1)/2})$ such that $\conv(S'')$ is an absolute $O(\eps)$-approximation of $\conv(S)$.

We then apply the polar transformation to the points of $S''$ yielding a set of $O(1/\eps^{(d-1)/2})$ halfspaces in the dual space. It follows from standard properties of the polar transformation that the polytope $\widehat{P}$ defined by the intersection of these halfspaces is fat and has constant diameter (see, e.g., Lemma~{7.2} of the journal version of \cite{AFM11}). This can be performed in time $O(1/\eps^{(d-1)/2})$.

Next, take a sufficiently large hypercube centered at the origin that contains $\widehat{P}$ and there is constant separation between the boundary of this hypercube and $\widehat{P}$. (Side length $O(1/\gamma)$ suffices.) We superimpose a grid of side length $\Theta(\sqrt{\eps})$ on each of the $2 d$ facets of this hypercube. Letting $G$ denote the resulting set of grid points on the boundary of the hypercube, we have $|G| = O(1/\eps^{(d-1)/2})$. Through the use of quadratic programming we compute the nearest neighbor of each point of $G$ on the boundary of $\widehat{P}$. For each grid point this can be done in time linear in the number of halfspaces that define $\widehat{P}$ \cite{Meg83,Meg84}. Thus, the total time for computing all the nearest neighbors is $O(|S''| \cdot |G|) = O(1/\eps^{d-1})$. Letting $\widehat{S}$ denote the set of nearest neighbors so obtained, we have $|\widehat{S}| = O(1/\eps^{(d-1)/2})$. By standard results, $\conv(\widehat{S})$ is an absolute $O(\eps)$-approximation of $\widehat{P}$ \cite{Dud74,BrI76}.

We again apply the polar transformation, mapping the set $\widehat{S}$ back to the primal space to obtain a set $H$ of $O(1/\eps^{(d-1)/2})$ halfspaces. Let $P''$ be the polytope formed by intersecting these halfspaces. Recalling that $\conv(\widehat{S})$ is an absolute $O(\eps)$-approximation of $\widehat{P}$, it follows from standard results that $P''$ is an absolute $O(\eps)$-approximation of $\conv(S'')$. This step takes time $O(1/\eps^{(d-1)/2})$. 

Since $P''$ is an absolute $O(\eps)$-approximation of $\conv(S'')$, and $\conv(S'')$ is an absolute $O(\eps)$-approximation of $P$, it follows that (subject to a suitable choice of constant factors) $P''$ is an absolute $O(\eps)$-approximation of $P$. Define $P'$ to be the polytope obtained by first translating the bounding halfspaces of $P''$ (i.e., the halfspaces of $H$) towards the origin by an amount $\Theta(\eps)$ and then intersecting the resulting halfspaces. Clearly, $P'$ is then an absolute inner $O(\eps)$-approximation of $P$, as desired. 

The overall running time is dominated by the time needed to compute the kernel, and the time needed to compute the nearest neighbors for the points of $G$.
\end{proof}

The following lemma is useful when representing polytopes by the intersection of halfspaces.  

\begin{lemma} \label{lem:conversion2}
Let $\gamma < 1$ be a positive constant, and $\eps > 0$ be a real parameter. Let $P$ be a polytope in $\gamma$-canonical form represented as the intersection of $n$ halfspaces. In $O(n + 1/\eps^{d-1})$ time it is possible to compute a polytope $P'$ represented as the intersection of $O(1/\eps^{(d-1)/2})$ halfspaces such that $P'$ is an outer absolute $\eps$-approximation of $P$.
\end{lemma}

\begin{proof}
Let $H$ denote the set of $n$ halfspaces defining $P$. We apply the polar transformation to the halfspaces of $H$ obtaining a set $S$ of $n$ points in the dual space. It follows from standard properties of the polar transformation that $\conv(S)$ is fat and has constant diameter (see, e.g., the journal version of \cite{AFM11}). This step can be performed in $O(n)$ time. By applying Chan's kernel construction~\cite{Cha06}, in time $O(n + 1/\eps^{d-1})$, we can compute a point set $S'$ of size $O(1/\eps^{(d-1)/2})$ such that $\conv(S')$ is an inner absolute $O(\eps)$-approximation of $\conv(S)$. We again apply the polar transformation, mapping the set $S'$ back to the primal space to obtain a set $H'$ of $O(1/\eps^{(d-1)/2})$ halfspaces. Let $P'$ be the polytope formed by intersecting these halfspaces. Since $\conv(S')$ is an inner absolute $O(\eps)$-approximation of $\conv(S)$, it follows that (subject to a suitable choice of constant factors), $P'$ is an absolute outer $O(\eps)$-approximation of $P$, as desired. The total time is dominated by the time to compute the kernel.
\end{proof}

\noindent \emph{Remark:} Theorem~\ref{thm:kernel} shows that an $\eps$-kernel of size $O(1/\eps^{(d-1)/2})$ can be computed in time $O(n \log \inv \eps + 1/\eps^{(d-1)/2 + \alpha})$.
The construction time in Lemma~\ref{lem:conversion2} is asymptotically the same as  the time needed to construct an $\eps$-kernel. Therefore, the construction time can be reduced to this quantity.

\section{Hierarchy of Macbeath Ellipsoids} \label{s:hierarchy}

The data structure presented in~\cite{AFM17a} for the approximate polytope membership problem is based on constructing a hierarchy of ellipsoids. In this section, we present a variant of this structure, which will play an important role in our constructions. 

For a Macbeath region $M^{\lambda}(x)$, we denote its circumscribing John ellipsoid by $E^{\lambda}(x)$, which we call a \emph{Macbeath ellipsoid}. Since Macbeath regions are centrally symmetric and the constant in John's Theorem~\cite{John} is $\sqrt{d}$ for centrally symmetric bodies, we have $E^{\lambda}(x) \subseteq M^{\lambda \sqrt{d}}(x)$. Recall the constant $\Delta_0 = \frac{1}{2}(\gamma^2 / 4d)^d$ defined in the statement of Lemma~\ref{lem:bound-mac}, and define $\lambda_0 =  1/(20d)$.
Each level of our structure is based on the following lemma. 
(We caution the reader that in the lemmas of this section, the value of $n$ used in the application of the lemma may differ from the original input size.) 

\begin{lemma} \label{lem:level}
Let $\gamma < 1$ be a positive constant, and let $0 < \delta \le \Delta_0$ be a real parameter. Let $P$ be a polytope in $\gamma$-canonical form, represented as the intersection of $n$ halfspaces.  There exists a set $X \subseteq \partial P(\delta)$ consisting of $O(1/\delta^{(d-1)/2})$ points such that the following properties hold:
\begin{enumerate}
\item[$(a)$\hspace{-5pt}] The set of Macbeath regions $\{M^{\lambda_0}(x) : x \in X\}$ are pairwise disjoint.
\item[$(b)$\hspace{-4pt}] The set of Macbeath ellipsoids $\{E^{4 \lambda_0 \sqrt{d}}(x) : x \in X\}$ together cover $\partial P(\delta)$.
\end{enumerate}
Furthermore, in  $O(n/\delta^{d-1} + 1/\delta^{3(d-1)/2})$ time, we can construct the set of Macbeath ellipsoids 
$\{E^{4 \lambda_0 \sqrt{d}}(x) : x \in X\}$.
\end{lemma}

\begin{proof}
We first show how to construct the required set of Macbeath ellipsoids. Translate each bounding halfspace of $P$ towards the origin by amount $\delta$. It is easy to see that the polytope $P(\delta)$ is the intersection of the translated halfspaces. This can be done in $O(n)$ time.

Recalling that  $P \subseteq B_0$, where $B_0$ is the ball of radius $r_0 = 1/2$, consider the hypercube just enclosing $B_0$. Superimpose a $\Theta(\delta)$-grid on each of the $2d$ facets of this hypercube. Intersect the segment joining the origin to each grid point with $\partial P(\delta)$, and let $X' \subset \partial P(\delta)$ denote the resulting set of intersection points. Note that $|X'| = O(1/\delta^{d-1})$. Using the fact that $P(\delta)$ is fat, a straightforward geometric calculation shows that for any point on $\partial P(\delta)$, there is a point of $X'$ within distance $c \delta$ of it, where $c$ is a suitable constant. (Adjusting the constant factor in the grid spacing, we can ensure that $c \le \lambda_0 \sqrt{d}$, which is a fact that we will use later in the proof.) As each point of $X'$ can be determined in $O(n)$ time, $X'$ can be computed in $O(n/\delta^{d-1})$ time.

For each $x' \in X'$, construct $M^{\lambda_0}(x')$. Let $\mathcal{M}'$ denote the resulting set of Macbeath regions. As a straightforward consequence of the definition of Macbeath regions, we can compute each Macbeath region in $O(n)$ time (i.e., in time proportional to the number of halfspaces that define $P$). Note that we represent each Macbeath region as the intersection of $n$ halfspaces. For each Macbeath region $M^{\lambda_0}(x') \in \mathcal{M}'$, determine the circumscribing John ellipsoid $E^{\lambda_0}(x')$. By standard results, we can construct the John ellipsoid of a convex polytope in time that is linear in the number of its defining halfspaces \cite{ChM96}. Thus, this step also takes time $O(n |\mathcal{M}'|) = O(n/\delta^{d-1})$.

Next, we will determine a maximal subset $\mathcal{M} \subseteq \mathcal{M}'$ such that the John ellipsoids associated with the Macbeath regions of $\mathcal{M}$ are disjoint. Initialize $\mathcal{M} = \emptyset$. Examine the Macbeath regions of $\mathcal{M}'$ one by one. Insert the Macbeath region into $\mathcal{M}$ if its associated John ellipsoid does not intersect the John ellipsoid of any Macbeath region of $\mathcal{M}$.  Clearly, this method yields a maximal subset $\mathcal{M} \subseteq \mathcal{M}'$ such that the associated John ellipsoids are disjoint. To bound the time required for this step, observe that the Macbeath regions of $\mathcal{M}$ are disjoint, and so by Lemma~\ref{lem:bound-mac}, $|\mathcal{M}| = O(1/\delta^{(d-1)/2})$. Since it takes constant time to check whether two ellipsoids intersect, it follows that the time required is $O(|\mathcal{M}'| \cdot |\mathcal{M}|) = O(1/\delta^{3(d-1)/2})$.

Finally, to obtain the desired ellipsoids let $X$ denote the set of centers of the Macbeath regions of $\mathcal{M}$. For each $x \in X$, we scale the associated ellipsoid $E^{\lambda_0}(x)$ constructed above about its center by a factor of $4 \sqrt{d}$ to obtain the ellipsoid $E^{4\lambda_0 \sqrt{d}}(x)$. This step can be done in time $O(|\mathcal{M}|) = O(1/\delta^{(d-1)/2})$. 

By combining the time of the above steps we obtain the desired overall construction time of $O(n/\delta^{d-1} + 1/\delta^{3(d-1)/2})$. The bound on $|X|$ follows from the bound on $\mathcal{M}$. By our earlier remarks, the set of Macbeath regions $\{M^{\lambda_0}(x) : x \in X\}$ are pairwise disjoint, which proves Property~(a). 

It remains to establish Property~(b), that is, to show that the union of the ellipsoids $\{E^{4\lambda_0 \sqrt{d}}(x) : x \in X\}$ covers $\partial P(\delta)$. Towards these end, consider any point $p \in \partial P(\delta)$. We will show that there is an $x \in X$ such that the ellipsoid $E^{4\lambda_0 \sqrt{d}}(x)$ contains $p$. Recall that there is a point $x'  \in X'$ that is within distance $c \delta$ of $p$, where $c$ is a constant no more than $\lambda_0 \sqrt{d}$. A straightforward consequence of the definition of Macbeath regions is that, for any point $y \in P$, the Macbeath region $M(y)$ contains a ball of radius $\delta(y)$ centered at $y$. It follows that $M^{\lambda_0\sqrt{d}}(x')$ contains a ball of radius $\lambda_0 \sqrt{d} \, \delta$ centered at $x'$. Hence, $p \in M^{\lambda_0\sqrt{d}}(x') \subseteq M^{4\lambda_0\sqrt{d}}(x') \subseteq E^{4\lambda_0\sqrt{d}}(x')$. Thus, if $x' \in X$, we are done. 

We next consider the case when $x' \notin X$. In this case, it follows from our construction that there is an $x \in X$ such that $E^{\lambda_0}(x)$ intersects $E^{\lambda_0}(x')$. Recall that $E^{\lambda_0}(x) \subseteq M^{\lambda_0 \sqrt{d}}(x)$ and $E^{\lambda_0}(x') \subseteq M^{\lambda_0 \sqrt{d}}(x')$. Thus $M^{\lambda_0 \sqrt{d}}(x)$ intersects $M^{\lambda_0 \sqrt{d}}(x')$. Applying Lemma~\ref{lem:mac-mac}, it follows that $M^{\lambda_0 \sqrt{d}}(x') \subseteq M^{4\lambda_0\sqrt{d}}(x)$. Thus $p \in M^{4\lambda_0\sqrt{d}}(x) \subseteq E^{4\lambda_0\sqrt{d}}(x)$. It follows that the ellipsoids $E^{4\lambda_0 \sqrt{d}}(x)$, for $x \in X$, together cover $\partial P(\delta)$, which completes the proof of the lemma.
\end{proof}

Based on the above lemma, we are now ready to describe our hierarchical data structure. Let $P$ be a polytope in $\gamma$-canonical form, where $\gamma$ is a constant. Recall the constant $\Delta_0 = \frac{1}{2} (\gamma^2 / (4d))^d$, and for $i \ge 0$ define $\Delta_i = \Delta_0 / 2^i$. The data structure consists of levels $0, 1, \ldots, \ell$, where $\ell$ is the smallest integer such that $\Delta_{\ell} \le \delta$. Since $\gamma$ is a constant, $\ell = O(\log \inv{\delta})$. Since $P$ is in $\gamma$-canonical form, the origin $O$ is at distance at least $\gamma/2$ from $\partial P$. Since $\Delta_0 < \gamma/2$, it follows that $P(\Delta_0)$ contains $O$ and $P(\Delta_i) \subset P(\Delta_{i+1})$ for $0 \le i \le \ell-1$. For each level $i$, we apply Lemma~\ref{lem:level}, setting $\delta$ in the lemma to $\Delta_i$. In $O(n/\Delta_i^{d-1} + 1/\Delta_i^{3(d-1)/2})$ time, we obtain a set of $O(1/\Delta_i^{(d-1)/2})$ ellipsoids that cover $\partial P(\Delta_i)$. Summing over all levels, the number of ellipsoids is $O(1/\delta^{(d-1)/2})$ and the time taken is $O(n/\delta^{d-1} + 1/\delta^{3(d-1)/2})$.

Our data structure is a directed acyclic graph (DAG), where the nodes at level $i$ correspond to the ellipsoids computed for level $i$. The children of an ellipsoid $E$ at level $i$ are the ellipsoids $E'$ at level $i+1$ such that there exists a ray from the origin that simultaneously intersects $E$ and $E'$. Since this is a constant time operation for any two given ellipsoids, it takes $O(1/\delta^{d-1})$ time to find the children of all the nodes in the DAG. It is convenient to root the DAG by creating a special node whose children are all the nodes of level zero. An important property of the hierarchy is that each node only has $O(1)$ children. The proof of this property is similar to that given in~\cite{AFM17a}. We include the proof for the sake of completeness. For the root, this follows from the fact that the number of nodes of level zero is  $O(1/\Delta_0^{(d-1)/2})$ and $\Delta_0$ is a constant. For non-root nodes, the proof is based on the following lemma, which is proved in~\cite{AFM17a}. 

\begin{lemma} \label{lem:degree}
Let $K \subset \RE^d$ be a convex body in $\gamma$-canonical form for some constant $\gamma$, and let $\lambda \le 1/5$ be any constant. Let $x \in K$ such that $\delta(x) \le \Delta_0$. Consider the generalized cone formed by rays emanating from the center $O$ of $K$ and intersecting $M'(x)$. Let $Y$ denote any set of points $y$ such that $\delta(y) = \delta(x)/2$ and the set of Macbeath regions $M^{\lambda}(y)$ are disjoint. Let $Y' \subseteq Y$ denote the set of points $y$ such that $M'(y)$ overlaps the aforementioned cone. Then $|Y'| = O(1)$.
\end{lemma}

For any node $w$, we let $x_w$ denote the center of the associated ellipsoid. Consider any node $u$ at level $i \ge 0$. Also, consider the generalized cone formed by rays emanating from the origin that intersect the ellipsoid $E^{4\lambda_0\sqrt{d}}(x_u)$ associated with $u$. The children of $u$ are those nodes $v$ at level $i+1$ whose ellipsoid $E^{4\lambda_0\sqrt{d}}(x_v)$ intersects this generalized cone. 

Since $x_u \in \partial P(\Delta_i)$, we have $\delta(x_u) = \Delta_i \le \Delta_0$. Recall that $E^{4\lambda_0\sqrt{d}}(x_u) \subseteq M^{4\lambda_0 d}(x_u) = M'(x_u)$. Thus, the generalized cone of rays that intersect $M'(x_u)$ includes all the rays used to define the children of $u$. The points $x_v$ that form level $i+1$ of the structure lie on $\partial P(\Delta_{i+1})$ and thus satisfy $\delta(x_v) = \delta(x_u) / 2$. By Property~(a) of Lemma~\ref{lem:level}, the Macbeath regions $M^{\lambda_0}(x_v)$ are pairwise disjoint, thus they constitute a set $Y$ as described in the preconditions of Lemma~\ref{lem:degree}. Each child $v$ of $u$ corresponds to a point $x_v$ such that the ellipsoid $E^{4\lambda_0\sqrt{d}}(x_v)$ intersects the generalized cone. Reasoning as we did above for $x_u$, we have $E^{4\lambda_0\sqrt{d}}(x_v) \subseteq M'(x_v)$. Therefore, the points $x_v$ associated with the children of $u$ constitute a subset of the set $Y'$ given in the lemma. Therefore, the number of children of $u$ is $O(1)$, as desired.

Given a query ray $Oq$, our data structure allows us to quickly find a leaf node such that the associated ellipsoid intersects this ray. The query algorithm descends the DAG by starting at the root and visiting any node at level zero that intersects the ray. Letting $u$ denote the current node, we next visit any child of $u$ whose associated ellipsoid intersects the ray. We repeat this procedure until a leaf node is reached. As the number of levels is $O(\log\inv\delta)$, this quantity bounds the time taken by this procedure. We summarize the main result of this section.

\begin{lemma} \label{lem:hierarchy}
Let $\gamma < 1$ be a positive constant, and let $0 < \delta \le \Delta_0$ be a real parameter. Let $P$ be a polytope in $\gamma$-canonical form, represented as the intersection of $n$ halfspaces. In  $O(n/\delta^{d-1} + 1/\delta^{3(d-1)/2})$ time, we can construct
the DAG structure described above. In particular, the DAG satisfies the following properties: 
\begin{enumerate}
\item[$(a)$\hspace{-5pt}] The total number of nodes (including leaves), and the total space used by the DAG are both $O(1/\delta^{(d-1)/2})$. 
\item[$(b)$\hspace{-4pt}] Each leaf is associated with an ellipsoid $E^{4\lambda_0\sqrt{d}}(x)$, where $x \in \partial P(\delta)$. The union of the
ellipsoids associated with all the leaves covers $\partial P(\delta)$.
\item[$(c)$\hspace{-4pt}] Given a query ray $Oq$, in $O(\log\inv\delta)$ time, we can find a leaf node such that the associated ellipsoid
intersects this ray.
\end{enumerate}
\end{lemma}

Given a convex body $K$ and query point $q$, an \emph{absolute $\eps$-APM} query returns a positive result if $q$ lies within $K$, a negative result if $q$ is at distance at least $\eps$ from $K$, and otherwise it may return either result. After a small enhancement, this DAG can be used for answering absolute $\eps$-APM queries for a polyope $P$ in $\gamma$-canonical form. We assume that $P$ is represented as the intersection of a set $H$ of $n$ halfspaces. We invoke the above lemma for $\delta = \eps \gamma / (2c_1)$, where $c_1$ is the constant of Lemma~\ref{lem:shadow}(b). We then associate each leaf of the DAG with a halfspace as follows. Let $x$ denote the center of the leaf ellipsoid and let $p$ denote the intersection of the ray $Ox$ with $\partial P$. Let $h \in H$ denote any supporting halfspace of $P$ (containing $P$) at $p$. We store $h$ with this leaf. By exhaustive search, we can determine $h$ in $O(n)$ time, so the total time for this step is $O(n / \eps^{(d-1)/2})$. Asymptotically, this does not affect the time it takes to construct the data structure. Given a query point $q$, we answer queries by first determining a leaf whose ellipsoid intersects the ray $Oq$. By Lemma~\ref{lem:hierarchy}(c), this takes $O(\log \inv\eps)$ time. We return a positive answer if and only if $q$ is contained in the associated halfspace.
We establish the correctness of this method in the following lemma.

\begin{lemma} \label{lem:APM-fat-corr}
Given a query point $q$, the query procedure returns a valid answer to the absolute $\eps$-APM query.
\end{lemma}

\begin{proof}
Consider the leaf whose ellipsoid intersects the ray $Oq$. Let $E^{4 \lambda_0 \sqrt{d}}(x)$ denote the associated ellipsoid and let $h$ be the halfspace stored with this leaf.  Recall that $h$ is a supporting halfspace at the point $p$ where the ray $Ox$ intersects $\partial P$. If $q \in P$ then clearly $q \in h$ and such a query point is correctly declared as lying inside $P$. To complete the proof, we need to show that if $q \notin P$ and the distance of $q$ from $\partial P$ is greater than $\eps$, then $q \notin h$. In this case, $q$ would be correctly declared as lying outside $P$.

Let $y$ denote any point in the intersection of the ray $Oq$ with the leaf ellipsoid. Let $y'$ denote the intersection of the ray $Oq$ with the hyperplane bounding $h$. To prove the claim, it suffices to show that $\|yy'\| \le \eps$. Recall that $E^{4 \lambda_0 \sqrt{d}}(x) \subseteq M^{4 \lambda_0 d}(x) = M'(x)$. Let $M = M'(x), S = \Sh(M)$, and $N = \normals(M)$. Clearly $y, p \in S$ and the normal vector $v$ to the hyperplane bounding $h$ belongs to $N$. By Lemma~\ref{lem:shadow}, $\width_v(S) \le c_1 \delta(x) = \eps \gamma / 2$. Note that the distance of $y$ from the hyperplane bounding $h$ is at most $\width_v(S)$. Applying Lemma~\ref{lem:raydist}(b), we obtain $\|yy'\| \le (2 / \gamma)(\eps \gamma / 2) = \eps$, as desired.
\end{proof}

We summarize the result below. 

\begin{lemma} \label{lem:APM-fat}
Let $\gamma < 1$ be a positive constant, and let $\eps > 0$ be a real parameter. Let $P$ be a polytope in $\gamma$-canonical form, represented as the intersection of $n$ halfspaces. In  $O(n/\eps^{d-1} + 1/\eps^{3(d-1)/2})$ time, we can construct a data structure that uses $O(1/\eps^{(d-1)/2})$ space and answers absolute $\eps$-APM queries in $O(\log\inv\eps)$ time.
\end{lemma}

\section{Kernel Construction} \label{s:kernel}

In this section we show how to build an $\eps$-kernel efficiently, proving Theorem~\ref{thm:kernel}. The input to an $\eps$-kernel construction consists of the approximation parameter $\eps$ and a set $S$ of $n$ points. Our algorithm is based on a bootstrapping strategy. We assume that we have access to an algorithm that can construct an $\eps$-kernel of $O(1/\eps^{(d-1)/2})$ size in time $O(n \log\inv\eps + 1/\eps^{(1/2 + \beta)(d-1)})$, where $\beta > 0$ is a parameter. Recall that the size of the kernel is asymptotically optimal in the worst case. We will present a method for improving the running time of this algorithm. Recall that Chan~\cite{Cha06} gave an algorithm for constructing kernels of  optimal size which runs in time $O(n \log \inv \eps + 1/\eps^{d-1})$. We will use this algorithm to initialize our bootstrapping scheme with $\beta = 1/2$.

Our method is based on executing the following steps. It uses a parameter $\delta = \eps^{1/3}$.
\begin{enumerate}
\item We begin by ``fattening'' the input point set $S$. Formally, we compute an affine transformation that maps $S$ to $S'$, such that $\conv(S')$ is in $\gamma$-canonical form for some constant $\gamma$. By standard results (see, e.g., the journal version of \cite{AFM11}), this affine transformation and the set $S'$ can be computed in $O(n)$ time.

\item Use Lemma~\ref{lem:conversion} to build a polytope $P$, represented as the intersection of $O(1/\delta^{(d-1)/2})$ halfspaces, such that $P$ is an inner absolute $\delta$-approximation of $\conv(S')$. This step takes $O(n + 1/\delta^{d-1}) = O(n + 1/\eps^{(d-1)/3})$ time.

\item Construct the DAG structure of Lemma~\ref{lem:hierarchy} for polytope $P$ using the parameter $\delta$. Replacing $n$ in the statement of the lemma by $O(1/\delta^{(d-1)/2})$, it follows that this step takes $O(1/\delta^{3(d-1)/2}) = O(1/\eps^{(d-1)/2})$ time.

\item For each point $p \in S'$, in $O(\log\inv\delta)$ time, we find a leaf of the DAG such that the associated ellipsoid $E^{4\lambda_0\sqrt{d}}(x)$ intersects the ray $Op$. Recall that $x \in \partial P(\delta)$. In $O(1)$ additional time, we can determine whether $p$ lies in the shadow of this ellipsoid (with respect to $\conv(S')$). If so, we associate $p$ with this ellipsoid, otherwise we discard it.  By Lemma~\ref{lem:hierarchy}(c), it takes $O(\log\inv\delta)$ time to process each point, thus the time taken for processing all the points of $S'$ is $O(n \log\inv\delta) = O(n \log\inv\eps)$.

\item For each leaf ellipsoid of the DAG, we build a $(c_3 \eps/\delta)$-kernel for the points of $S'$ that lie in its shadow, where $c_3$ is a suitably small constant that will be selected later. This kernel computation is done using the aforementioned algorithm that computes $\eps$-kernels of point sets of size $n$ in time $O(n \log\inv\eps + 1/\eps^{(1/2+\beta)(d-1)})$. The size of the $O(\eps/\delta)$-kernel computed for each shadow is $O((\delta/\eps)^{(d-1)/2})$ and the time required is $O(n_i \log\frac{\delta}{\eps} + (\delta/\eps)^{(1/2+\beta)(d-1)})$, where $n_i$ denotes the number of points of $S'$ in the shadow. Summed over all the shadows, it follows that the total time required is
\[
O\left( n \log \frac{\delta}{\eps}  + \left(\frac{1}{\delta}\right)^{\kern-2pt \frac{d-1}{2}} \kern-2pt \left( \frac{\delta}{\eps} \right)^{\kern-2pt \left(\inv 2 + \beta\right)(d-1)} \right) =
O\left( n \log \frac{1}{\eps}  + \left( \inv \eps \right)^{\kern-2pt \left(\inv 2 + \frac{2\beta}{3}\right)(d-1)}  \right).
\]
Here we have used the facts that each point of $S'$ is assigned to at most one shadow and the total number of shadows, which is bounded by the number of leaves in the DAG, is $O(1/\delta^{(d-1)/2})$.

\item Let $S'' \subseteq S'$ be the union of the kernels computed in the previous step.  Since the number of shadows is $O(1/\delta^{(d-1)/2})$ and the size of the kernel for each shadow is  $O((\delta/\eps)^{(d-1)/2})$, it follows that $|S''| = O(1/\eps^{(d-1)/2})$. We apply the inverse of the affine transformation computed in Step 1 to the points of $S''$, and output the resulting set of points as the desired $\eps$-kernel for $S$.
\end{enumerate}

We have shown that the size of the output kernel is $O(1/\eps^{(d-1)/2})$, as desired. The running time of Step 5 dominates the time complexity.
The next lemma establishes the correctness of this construction.

\begin{lemma} \label{lem:correctness}
The construction yields an $\eps$-kernel.
\end{lemma}

\begin{proof}
Throughout this proof, for a given convex body $K$, we use $M_K(x)$, $E_K(x)$, and $\delta_K(x)$ to denote the quantities $M(x)$, $E(x)$, and $\delta(x)$ with respect to $K$. Let $P' = \conv(S')$. By standard results on fattening, it suffices to show that $\conv(S'')$ is an absolute $O(\eps)$-approximation of $P'$. Let $v$ be an arbitrary direction. Consider the extreme point $p$ of $S'$ in direction $v$. Clearly $p \in \partial P'$. Recall that $P$ is an inner $\delta$-approximation of $P'$, and the ellipsoids associated with the leaves of the DAG cover the boundary of $P(\delta)$. Thus, there must be an ellipsoid $E = E_P^{4\lambda_0\sqrt{d}}(x)$, $x \in \partial P(\delta)$, such that $p$ is assigned to the shadow of $E$ in Step 4. Note that this shadow and all shadows throughout this proof are assumed to be with respect to the polytope $P'$ (and not $P$). We claim that $\width_v(\Sh(E)) \le 2 c_1 \delta$, where $c_1$ is the constant of Lemma~\ref{lem:shadow}(b). Assuming this claim for now, let us complete the proof of the lemma. Recall that in Step 5, we built a $(c_3 \eps/\delta)$-kernel for all the points of $S'$ that are assigned to the shadow of $E$, and $S''$ includes all the points of this kernel. It follows that the distance between the supporting hyperplanes of $\conv(S')$ and $\conv(S'')$ in direction $v$ is at most $(c_3 \eps / \delta) \cdot \width_v(\Sh(E)) \le (c_3 \eps / \delta) \cdot (2 c_1 \delta) = 2 c_1 c_3 \eps$. By choosing $c_3$ sufficiently small, we can ensure that this quantity is smaller than any desired constant times $\eps$, which proves the lemma.

It remains to show that $\width_v(\Sh(E)) \le 2 c_1 \delta$. Recall that 
\[
E 
	~ = ~ E_P^{4\lambda_0 \sqrt{d}}(x) 
	~ \subseteq ~ M_P^{4\lambda_0 d}(x) 
	~ = ~ M'_P(x).
\]
Furthermore, since $P \subseteq P'$, a straightforward consequence of the definition of Macbeath regions is that $M'_P(x) \subseteq M'_{P'}(x)$. To simplify the notation, let $M$ denote $M'_{P'}(x)$. Putting it together, we obtain $E \subseteq M$. Thus $\Sh(E) \subseteq \Sh(M)$, which implies that $\width_v(\Sh(E)) \le \width_v(\Sh(M))$. By Lemma~\ref{lem:shadow}(b),
\[
\width_v(\Sh(M)) 
	~ \le ~ c_1 \delta_{P'}(x). 
\]
Using the triangle inequality and the fact that $P$ is an inner $\delta$-approximation of $P'$, we obtain $\delta_{P'}(x) \le \delta_P(x) + \delta = 2 \delta$. Thus $\width_v(\Sh(E)) \le \width_v(\Sh(M)) \le 2 c_1 \delta$, as desired.
\end{proof}

We are now ready to prove Theorem~\ref{thm:kernel}.

\begin{proof}
Our proof is based on a constant number of applications of the algorithm from this section.
It suffices to show that there is an algorithm that can construct an $\eps$-kernel of $O(1/\eps^{(d-1)/2})$ size in time $O(n \log\inv\eps + 1/\eps^{(1/2 + \beta')(d-1)})$, where $\beta' = \alpha / (d-1)$. 

We initialize the bootstrapping process by Chan's algorithm~\cite{Cha06}, which has $\beta = 1/2$.
Observe that the value of $\beta$ is initially $1/2$ and falls by a factor of $2/3$ with each application of the algorithm.
It follows that after $O(\log \inv \alpha)$ applications, we will obtain an algorithm with the desired running time. This completes the proof.
\end{proof}

\section{Approximate Polytope Membership} \label{s:apm}

In this section we show how to obtain a data structure for approximate polytope membership, proving Theorem~\ref{thm:apm}. Our best data structure for APM achieves query time $O(\log\inv\eps)$ with storage $O(1/\eps^{(d-1)/2})$ and preprocessing time $O(n \log\inv\eps + 1/\eps^{(d-1)/2 + \alpha})$.
As with kernels, our construction here is again based on a bootstrapping strategy. To initialize the process, we will use a data structure that achieves the aforementioned query time with the same storage but with preprocessing time $O(n + 1/\eps^{3(d-1)/2})$. The data structure is based on Lemma~\ref{lem:APM-fat}. Recall that the input is a polytope represented as the intersection of $n$ halfspaces. 

We begin by ``fattening'' the input polytope. Formally, we use an affine transformation to map the input polytope to a polytope $P'$ that is in $\gamma$-canonical form. This step takes $O(n)$ time \cite{AFM11}. By standard results, it suffices to build a data structure for answering absolute $O(\eps)$-APM queries with respect to $P'$ (see, e.g., Lemma~{7.1} of the journal version of \cite{AFM11}). 

Next, we apply Lemma~\ref{lem:conversion2} to construct an outer absolute $O(\eps)$-approximation $P$ of $P'$, where $P$ is represented as the intersection of $O(1/\eps^{(d-1)/2})$ halfspaces. This step takes $O(n + 1/\eps^{d-1})$ time. Finally, we use Lemma~\ref{lem:APM-fat} to construct a data structure for answering absolute $O(\eps)$-APM queries with respect to $P$. Replacing $n$ in the statement of the lemma by $O(1/\eps^{(d-1)/2})$, it follows that this step takes $O(1/\eps^{3(d-1)/2}))$ time. 

The total construction time is $O(n + 1/\eps^{3(d-1)/2})$. To answer a query, we map the query point using the same transformation used to fatten the polytope, and then use the data structure constructed above to determine whether the resulting point lies in polytope $P$. Subject to an appropriate choice of constant factors, the correctness of this method follows from the fact that $P$ is an outer absolute $O(\eps)$-approximation of $P'$. 

We summarize this result in the following lemma.

\begin{lemma} \label{lem:APM1}
Let $\eps > 0$ be a real parameter and let $P$ be a polytope, represented as the intersection of $n$ halfspaces. In $O(n + 1/\eps^{3(d-1)/2})$ time, we can construct a data structure that uses $O(1/\eps^{(d-1)/2})$ space and answers $\eps$-APM queries in $O(\log\inv\eps)$ time.
\end{lemma}

We can now present the details of our bootstrapping approach. We assume that we have access to a data structure that can answer $\eps$-APM queries in $O(\log\inv\eps)$ time with $O(1/\eps^{(d-1)/2})$ storage and $O(n \log\inv\eps + 1/\eps^{(1/2 + \beta)(d-1)})$, where $\beta > 0$ is a parameter. We present a method for constructing a new data structure which matches the given data structure in space and query time, but has a lower preprocessing time. Our method uses a parameter $\delta = \eps^{\beta/(1+\beta)}$.
\begin{enumerate}
\item As in the construction given above, we first fatten the input polytope. Formally, we use an affine transformation to map the input polytope to a polytope $P'$ that is in $\gamma$-canonical form. This step takes $O(n)$ time. By standard results, it suffices to build a data structure for answering absolute $O(\eps)$-APM queries with respect to $P'$. 

\item Use Lemma~\ref{lem:conversion2} to construct an outer absolute $O(\eps)$-approximation $P$ of $P'$, where $P$ is represented as the intersection of $O(1/\eps^{(d-1)/2})$ halfspaces. By the remark following Lemma~\ref{lem:conversion2}, this step takes $O(n \log\inv\eps + 1/\eps^{(d-1)/2 + \alpha})$ time.

\item Construct the DAG of Lemma~\ref{lem:hierarchy} for polytope $P$ using the parameter $\delta$. Replacing $n$ in the statement of the lemma by $O(1/\eps^{(d-1)/2})$, it follows that this step takes $O((1/\delta)^{d-1} \cdot (1/\eps)^{(d-1)/2})$ time.

\item For each leaf of the DAG, we construct an APM data structure as follows. Let $E = E^{4 \lambda_0 \sqrt{d}}(x)$ denote the ellipsoid associated with the leaf. Let $R$ denote the minimum enclosing hyperrectangle of the ellipsoid $E^{4/\gamma}(x)$. We will see later that $R$ contains the shadow of $E$ (with respect to $P$), and its width in any direction in $\normals(E)$ is at most $c_2 d \delta = O(\delta)$, where $c_2$ is the constant in Lemma~\ref{lem:shadow}(c). We use the aforementioned algorithm for constructing an APM data structure for this region with approximation parameter $c_3 \eps / \delta$, where $c_3$ is a sufficiently small constant that we will select later. Note that each such region can be expressed as the intersection of $n_i = O(1/\eps^{(d-1)/2})$ halfspaces, namely, all the halfspaces defining $P$ together with the $2d$ halfspaces defined by the facets of $R$. The construction time of the APM data structure for each leaf is
\[
O\left( n_i \log \frac{\delta}{\eps} + \left( \frac{\delta}{\eps} \right)^{\kern-2pt \left(\inv 2 + \beta\right)(d-1)} \right)
	~ = ~ O\left( \left(\frac{1}{\eps}\right)^{\frac{d-1}{2}} \log \frac{\delta}{\eps} + \left( \frac{\delta}{\eps} \right)^{\kern-2pt \left(\inv 2 + \beta\right)(d-1)} \right),
\]
and the space used is $O((\delta / \eps)^{(d-1)/2})$. Since there are $O(1/\delta^{(d-1)/2})$ leaves, it follows that the total space is $O(1/\eps^{(d-1)/2})$, and the total construction time is the product of $O(1/\delta^{(d-1)/2})$ and the above construction time for each leaf.
\end{enumerate}

Summing up the time over all the four steps, we get a total construction time on the order of
\[
n \log \frac{1}{\eps} + \left( \frac{1}{\eps} \right)^{\frac{d-1}{2} + \alpha} \kern-8pt +
	\left( \frac{1}{\delta} \right)^{d-1} \left( \frac{1}{\eps} \right)^{\frac{d-1}{2}} +
	\left(\frac{1}{\delta}\right)^{\frac{d-1}{2}} \cdot \left( \left(\frac{1}{\eps}\right)^{\frac{d-1}{2}}  \kern-4pt \log \frac{\delta}{\eps}  + \left( \frac{\delta}{\eps} \right)^{\kern-2pt \left(\inv 2 + \beta\right)(d-1)} \right).
\]
Recalling that $\delta = \eps^{\beta/(1+\beta)}$ and assuming that the constant $\alpha$ is much smaller than $\beta$, it follows that the construction time is
\[
O\left( n \log \frac{1}{\eps}  + \left( \inv \eps \right)^{\kern-2pt \left(\inv 2 + \frac{\beta}{1+\beta}\right)(d-1)} \right).
\]

We answer queries as follows. Recall the affine transformation used to fatten the input polytope. We apply this transformation on the input query point to obtain a point $q$. Recall that it suffices to answer absolute $O(\eps)$-APM queries for $q$ with respect to $P'$. As $P$ is an outer absolute $O(\eps)$-approximation of $P'$, it suffices to answer absolute $O(\eps)$-APM queries for $q$ with respect to $P$. To answer this query, we identify a leaf of the DAG such that the associated ellipsoid $E$ intersects the ray $Oq$. This takes time $O(\log\inv\delta)$. Let $y$ denote an intersection point of this ray with the ellipsoid $E$. If $q$ lies on the segment $Oy$, then $q$ is declared as lying inside $P$. Otherwise we return the answer we get for query $q$ using the APM data structure we built for this leaf. It takes time $O(\log\frac{\delta}{\eps})$ to answer this query. Including the time to locate the leaf, the total query time is $O(\log\inv\eps)$. 

In Lemma~\ref{APM-corr}, we show that queries are answered correctly.

\begin{lemma} \label{APM-corr}
The query procedure returns a valid answer to the $\eps$-APM query.
\end{lemma}

\begin{proof}
We borrow the terminology from the query procedure given above. As mentioned, it suffices to show that our algorithm correctly answers absolute $O(\eps)$-APM queries for $q$ with respect to the polytope $P$. Recall that we identify a leaf of the DAG whose associated ellipsoid $E = E^{4\lambda_0 \sqrt{d}}(x)$ intersects the ray $Oq$. Recall that $y$ is a point on the intersection of the ray $Oq$ with $E$. Clearly, if $q$ lies on segment $Oy$, then $q \in P$ and $q$ is correctly declared as lying inside $P$.

It remains to show that queries are answered correctly when $\|Oq\| > \|Oy\|$. In this case, we handle the query using the APM data structure we built for the leaf. Recall that this structure is built for the polytope formed by intersecting $P$ with the smallest enclosing hyperrectangle $R$ of the ellipsoid $E^{4/\gamma}(x)$. We claim that (i) $\Sh(E) \subseteq R$ and (ii) $\width_v(R) \le c_2 d \delta$ for all $v \in \normals(E)$, where $c_2$ is the constant in Lemma~\ref{lem:shadow}(c).

To see this claim, recall that $M^{\lambda}(x) \subseteq E^{\lambda}(x) \subseteq M^{\lambda \sqrt{d}}(x)$ for any $\lambda > 0$. Using this fact, it follows that $M^{4/\gamma}(x) \subseteq E^{4/\gamma}(x) \subseteq  M^{4 \sqrt{d}/\gamma}(x)$. By Lemma~\ref{lem:shadow}(a), $\Sh(E) \subseteq M^{4/\gamma}(x)$. Thus $\Sh(E) \subseteq E^{4/\gamma}(x) \subseteq R$, which proves (i). To prove (ii), note that $R \subseteq E^{4 \sqrt{d}/\gamma}(x)$, since $R$ is the smallest enclosing hyperrectangle of $E^{4/\gamma}(x)$. Also $E^{4 \sqrt{d}/\gamma}(x) \subseteq M^{4 d/\gamma}(x)$. Thus $R \subseteq M^{4 d/\gamma}(x)$. By Lemma~\ref{lem:shadow}(c), $\width_v(M^{4/\gamma}(x)) \le c_2 \delta$ for all $v \in \normals(M'(x))$. Since $R \subseteq M^{4d / \gamma}(x)$ and $E \subseteq M'(x)$, it follows that $\width_v(R) \le c_2 d \delta$ for all $v \in \normals(E)$.

We return to showing that queries are correctly answered when $\|Oq\| > \|Oy\|$. We consider two possibilities depending on whether $q$ is inside or outside $P$. If $q \in P$ then $q \in \Sh(E)$. By part (i) of the above claim, $\Sh(E) \subseteq R$. Thus $q \in P \cap R$. It follows that the APM structure built for the leaf will declare this point as lying inside $P \cap R$, and hence the overall algorithm will correctly declare that $q$ lies in $P$.

Finally, we consider the case when $q \notin P$. To complete the proof, we need to show that if the distance of $q$ from the boundary of $P$ is greater than $\eps$, then $q$ is declared as lying outside $P$. Let $p$ denote the point of intersection of the ray $Oq$ with $\partial P$, let $h$ denote a hyperplane supporting $P$ at $p$, and let $v$ denote the outward normal to $h$. Recall by part (i) of the claim that $\Sh(E) \subseteq R$. It follows that $h$ is a supporting hyperplane of $P \cap R$ at $p$. By part (ii) of the claim, $\width_v(R) \le c_2 d \delta$. It follows that $\width_v(P \cap R) \le c_2 d \delta$. Recall that the APM data structure for the leaf is built using approximation parameter $c_3 \eps / \delta$ for some constant $c_3$. By definition of APM query (in the standard, direction-sensitive sense), the absolute error allowed in direction $v$ is at most $(c_3 \eps / \delta) \cdot \width_v(P \cap R)  \le (c_3 \eps / \delta) (c_2 d \delta)$. By choosing $c_3$ sufficiently small we can ensure that this error is at most $\eps \gamma$. To make this more precise, let $h'$ denote the hyperplane parallel to $h$ (outside $P$), and at distance $\eps \gamma$ from it. Consider the halfspace bounded by $h'$ and containing $P$. By definition of APM query, if $q$ is not contained in this halfspace, then $q$ would be declared as lying outside $P \cap R$, and the overall algorithm would declare $q$ as lying outside $P$. Let $p'$ denote the point of intersection of the ray $Oq$ with $h'$. By Lemma~\ref{lem:raydist}(c), $\|pp'\| \le (\eps \gamma)  / \gamma = \eps$. Thus, if the distance of $q$ from $\partial P$ is greater than $\eps$, then $q$ cannot lie on segment $pp'$ and $q$ is correctly declared as lying outside $P$. This completes the proof of correctness.
\end{proof}

We are now ready to prove Theorem~\ref{thm:apm}

\begin{proof}
Our proof is based on a constant number of applications of the method presented in this section.
It suffices to show that there is a data structure with space and query time as in the theorem and preprocessing time $O(n \log\inv\eps + 1/\eps^{(1/2 + \beta')(d-1)})$, where $\beta' = \alpha / (d-1)$. 

We initialize the bootstrapping process by the data structure described in the beginning of this section, which has $\beta = 1$. Recall that applying the method once changes the value of $\beta$ to $\beta/(1+\beta)$. It is easy to show that after $i$ applications, the
value of $\beta$ will fall to $1/(i+1)$. Thus, after $O(1/\alpha)$ applications, we will obtain a
data structure with the desired preprocessing time.
\end{proof}

\section{Reductions} \label{s:reductions}

In this section, we show how the remaining problems reduce to polytope membership. We start with a useful variation of approximate nearest neighbor searching.

The input for an approximate nearest neighbor searching data structure is a set $S$ of data points and an approximation parameter $\eps$. Given a constant $\sigma>0$, \emph{$\sigma$-well-separated} approximate nearest neighbor searching is defined as follows. Let $Q_S$ and $Q_q$ be two hypercubes of side length $r$ and at distance at least $\sigma r$ from each other. In the $\sigma$-well-separated version we have the data points $S$ inside $Q_S$ and the query points inside $Q_q$. Data structures for the well-separated version are much more efficient than for the unrestricted version. The following reduction from well-separated approximate nearest neighbor searching to approximate polytope membership is presented in~\cite[Lemma~9.2 of the journal version]{AFM11}.

\begin{lemma} \label{lem:wsann-reduction}
Let $0 < \eps \leq 1/2$ be a real parameter, $\sigma > 0$ be a constant, and $S$ be a set of $n$ points in $\RE^d$. Given a data structure for approximate polytope membership in $d$-dimensional space with query time $t_{d}(\eps)$, storage $s_{d}(\eps)$, and preprocessing time $O(n \log \inv \eps + b_d(\eps))$ it is possible to preprocess $S$ into a $\sigma$-well-separated ANN data structure with
\[
\hbox{Query time:~} O\kern-2pt \left(t_{d+1}(\eps) \cdot \log \inv{\eps}\right)\hspace{0.8em}
\textrm{Space:~} O\big(s_{d+1}(\eps) \big)\hspace{0.8em}
\textrm{Preprocessing:~} O\kern-2pt \left(\kern-2pt n \log \inv \eps +  b_{d+1}(\eps)\right)\kern-2pt.
\]
\end{lemma}

Combining the previous reduction with Theorem~\ref{thm:apm} we have:

\begin{lemma} \label{lem:wsann}
Given a set $S$ of $n$ points in $\RE^d$, an approximation parameter $\eps > 0$, and a constant $\sigma>0$, there is a data structure that can answer $\sigma$-well-separated Euclidean $\eps$-approximate nearest neighbor queries with
\[
\hbox{Query time:~} O\kern-2pt \left(\log^2 \inv\eps\right)\quad
\textrm{Space:~} O\kern-2pt \left( \left(\inv \eps \right)^{\kern-2pt \half{d}}\right)\quad
\textrm{Preprocessing:~} O\kern-2pt \left( n \log \inv \eps + \left(\inv\eps\right)^{\kern-2pt \half{d}+\alpha} \right).
\]
\end{lemma}

Next, we prove Theorem~\ref{thm:bcp} using a reduction to well-separated approximate nearest neighbor searching that is based on~\cite[Theorem~3.2]{ArC14}.

\begin{proof}
Let $b$ denote the exact BCP distance. We obtain a constant approximation $b \leq a < 2b$ of the BCP distance in $O(n)$ expected time by running the randomized algorithm from~\cite{KhM95}. Then, we build a grid with cells of diameter $a/4$ and partition the red points accordingly. Note that since $a/4 < b/2$, the BCP pair cannot be in the same grid cell, nor in two adjacent cells. The strategy of the algorithm is to partition the red points among the grid cells and to perform a constant number of well-separated approximate nearest neighbor queries for each blue point, returning the closest red-blue pair found. More precisely, for each blue point $q$, we perform an approximate nearest neighbor query among the grid cells $Q_S$ that intersect the set theoretic difference of two balls of radii $a$ and $a/2$ centered around $q$. These are the only grid cells that may contain the closest red point and, by a simple packing argument, the number of grid cells $Q_S$ is constant. Since the grid cell $Q_q$ that contains $q$ cannot be adjacent to $Q_S$, it follows that the separation $\sigma$ is at least $1$.

To answer the queries efficiently, we separate the grid cells onto two types. If the number of red points in the cell is greater than $1 / \eps^{d/4}$, we say the cell is \emph{heavy}, and otherwise we say the cell is \emph{light}. Clearly, the number of heavy cells is $O(n \cdot  \eps^{d/4})$. We build well-separated approximate nearest-neighbor data structures for the heavy cells.  Using Lemma~\ref{lem:wsann}, the total preprocessing time is $O(n / \eps^{d/4 + \alpha})$. For each light cell, we simply store the red points it contains and answer nearest neighbor queries by brute force in $O(1/\eps^{d/4})$ time. Therefore, the total time spent answering queries is $O(n / \eps^{d/4})$.
\end{proof}

An approximation to the Euclidean minimum spanning tree and minimum bottleneck tree can be computed by solving multiple BCP instances such that the sum of the number of points in all instances is $O(n \log n)$~\cite[Theorem 4.1]{ArC14}. Applying this reduction together with Theorem~\ref{thm:bcp}, we prove Theorem~\ref{thm:mst}.

The following reduction from ANN to APM is presented in~\cite[Lemma~9.3 of the journal version]{AFM11}. For the preprocessing time see \cite[Lemma~8.3]{AVD-JACM}.

\begin{lemma} \label{lem:ann-reduction}
Let $0 < \eps \leq 1/2$ be a real parameter and $S$ be a set of $n$ points in $\RE^d$. Given a data structure for approximate polytope membership in $d$-dimensional space with query time at most $t_{d}(\eps)$ and storage $s_{d}(\eps)$, and preprocessing time $O(n \log \inv \eps + b_d(\eps))$ it is possible to preprocess $S$ into an ANN data structure with
\[
\begin{aligned}
& \hbox{Query time:~} O\kern-2pt \left(\log n + t_{d+1}(\eps) \cdot \log \inv{\eps}\right)\quad
\textrm{Space:~} O\kern-2pt \left(n\,\log\inv{\eps} + n \, \frac{s_{d+1}(\eps)}{t_{d+1}(\eps)} \right) \\
& \textrm{Preprocessing:~} O\kern-2pt \left( n \log n \kern+1pt \log \inv{\eps} +  n \, \frac{b_{d+1}(\eps)}{t_{d+1}(\eps)}\right).
\end{aligned}
\]
\end{lemma}

Applying this reduction with the data structure from Theorem~\ref{thm:apm} and setting $t_{d+1}(\eps) = 1/(m \cdot \eps^{d/2})$ for $\log \inv{\eps} \leq m \leq 1/(\eps^{d/2}\log\inv\eps)$, we obtain Theorem~\ref{thm:ann}. 

Next, we show how to obtain a data structure for approximate directional width queries (Theorem~\ref{thm:widthqueries}) using the data structure for approximate polytope membership from Theorem~\ref{thm:apm}. The proof uses standard duality and binary search techniques.

\begin{proof}
Given a polytope $P$ (defined as the intersection of $n$ halfspaces) that contains the origin $O$, we define a ray-shooting query (from the origin) as follows. Let $v$ be a query direction and let $r$ denote the ray emanating from $O$ in direction $v$. The result of the query $q(P,v)$ is the length of $r \cap P$. In the $\eps$-approximate version, any answer between $q(P,v)$ and $(1+\eps) q(P,v)$ is acceptable.

If we place the origin $O$ in the center of the John ellipsoid of $P$, we have $q(P,-v) = \Theta(q(P,v))$ for all $v$. Thus, a constant approximation of $q(P,v)$ can be obtained by replacing $P$ by its circumscribing John ellipsoid. We can then refine the approximation using binary search and approximate polytope membership queries. (To see this, consider the point $p \in \partial P$ that is hit by the ray, and let $h$ be any supporting hyperplane at $p$. Consider the slab containing $P$ that is bounded by this hyperplane and the parallel hyperplane on the opposite side of $P$. By properties of the John ellipsoid, the origin lies within a central region of the slab. It follows from basic geometry that if we expand the slab by $\eps$ times its width, the ratio between ray distances to the expanded slab boundary and the original slab boundary is $1 + O(\eps)$. An $\eps$-APM query with respect to $P$ along this ray will achieve an approximation error that is no greater.) By a suitable adjustment to the constant factor, we can obtain an $\eps$-approximation to $q(P,v)$ after $O(\log \inv \eps)$ membership queries.

The polar body $P^*$ (defined as the convex hull of $n$ points) of $P$ has the property that $\width_v(P^*) = 1/q(P,v) + 1/q(P,-v)$. Therefore, we can $\eps$-approximate the width of a set of points $P^*$ using $O(\log \inv \eps)$ approximate polytope membership queries on $P$ and Theorem~\ref{thm:widthqueries} follows.
\end{proof}

Agarwal, Matou{\v{s}}ek, and Suri~\cite{AMS92} showed that the diameter of a point set $S$ can be $\eps$-approximated by computing the maximum width of $S$ among $O(1/\eps^{(d-1)/2})$ directions. Therefore, Theorem~\ref{thm:diameter} follows immediately from Theorem~\ref{thm:widthqueries}.

\pdfbookmark[1]{References}{s:ref}
\bibliographystyle{abbrv}
\bibliography{shortcuts,coreset}

\end{document}